\documentclass[11pt]{article}

\usepackage{amssymb,amsthm, amscd}
\usepackage{amsmath}
\usepackage{epsfig}
\usepackage{bm}
\usepackage{verbatim}

\newtheorem{lemma}{Lemma}

\newtheorem{theorem}[lemma]{Theorem}

\newtheorem{corollary}[lemma]{Corollary}

\newtheorem{remark}[lemma]{Remark}

\newtheorem{claim}[lemma]{Claim}

\newcommand{\lf}{\lfloor}
\newcommand{\rf}{\rfloor}
\newcommand{\lc}{\lceil}
\newcommand{\rc}{\rceil}

\newcommand{\HH}{\mathcal{H}}
\newcommand{\II}{\mathcal{I}}

\newcommand{\calS}{\mathcal{S}}

\newcommand{\ZZ}{\mathcal{Z}}

\newcommand{\CN}{\mathcal{CN}}

\renewcommand{\phi}{\varphi}

\newcommand{\C}{{\mathbb{C}}}

\newcommand{\inv}{^{-1}}

\newcommand{\beq}{\begin{equation}}
\newcommand{\eeq}{\end{equation}}

\newcommand{\Id}{\mathrm{Id}}

\newcommand{\lrf}[1]{\left\lf #1 \right\rf}
\newcommand{\lrc}[1]{\left\lc #1 \right\rc}
\newcommand{\Epq}{{E_{pq}}}
\newcommand{\Eqp}{{E_{qp}}}
\newcommand{\tr}{\mathrm{Tr}}
\newcommand{\conj}{\mathrm{conj}}
\newcommand{\dc}{\lrc{\frac d2}}
\newcommand{\df}{\lrf{\frac d2}}

\author{Guy Bresler
\thanks{Guy Bresler is with Wireless Foundations,
Dept of EECS, UC Berkeley, Berkeley, CA 94720,
Email: \textsf{gbresler@eecs.berkeley.edu}.
}
\and
Dustin Cartwright
\thanks{Dustin Cartwright is with
Institut Mittag-Leffler, Djursholm, Sweden
Email: \textsf{dustin@math.berkeley.edu}}
\and
David Tse
\thanks{David Tse is with Wireless Foundations,
Dept of EECS, UC Berkeley, Berkeley, CA 94720,
Email: \textsf{dtse@eecs.berkeley.edu}}}

\title{Settling the feasibility of interference alignment for the MIMO interference channel: the symmetric square case}
\begin{document}
\date{}
\maketitle
\begin{abstract}

Determining the feasibility conditions for vector space interference alignment in the $K$-user MIMO interference channel with constant channel coefficients has attracted much recent attention yet remains unsolved.
The main result of this paper is restricted to the symmetric square case where all transmitters and receivers have $N$ antennas, and each user desires $d$ transmit dimensions. We prove that alignment is possible if and only if the number of antennas satisfies $N\geq d(K+1)/2$. We also show a necessary condition for feasibility of alignment with arbitrary system parameters.
An algebraic geometry approach is central to the results.
\end{abstract}

\section{Introduction}
Interference alignment has inspired much hope as the savior to the problem of
communication in the presence of interference. Introduced by Maddah-Ali et al.
\cite {MMK08} for the multiple-input multiple-output (MIMO) X channel and subsequently by
Cadambe and Jafar \cite{CJ08} in the context of the $K$-user interference
channel (IC), the basic idea is to align multiple interfering signals at each
receiver in order to reduce the effective interference. For the $K$-user IC, in
the case of independently faded parallel channels (i.e. time or frequency
selective), it was shown in \cite{CJ08} that up to $K/2$ total
degrees-of-freedom is achievable: amazingly, this implies that each user gets
the same degrees of freedom as in a simple 2-user IC. However, the result
depends critically on the assumption that the number of independently faded
parallel channels, i.e. the channel diversity, is \emph{unbounded} and in fact grows like $K^{2K^2}$. A physical
system has only a \emph{finite} channel diversity, which raises the question of how much interference alignment is possible if there is some fixed---finite---number of parallel channels. This problem was addressed for the 3-user channel in \cite{BT09}.

In this paper we consider the $K$-user MIMO IC, where each transmitter and receiver has multiple antennas, but the channel is constant over time and frequency. Similar in flavor to the situation with finite time or frequency diversity in \cite{BT09}, here we have a fixed amount of spatial diversity due to the multiple antenna elements, and the goal is to design the best communication strategy for the system at hand. We restrict attention to vector space strategies, both due to their relative tractability and because we believe the essence of the interference problem remains.
As discussed subsequently, the problem of maximizing the degrees-of-freedom is equivalent to determining whether it is possible to align interfering signals (termed \emph{the alignment problem}), given the system parameters and desired signal dimensions $d_i$ for  $1\leq i\leq K$. Our main result is restricted to the fully symmetric case with $N$ antennas at each transmitter and receiver, and $d$ desired signalling dimensions per user. In this setting, we completely characterize the feasibility of interference alignment, informally stated as follows:
\begin{theorem}\label{t:symmetricFeasibility}
Fix the number of users $K$, number of antennas $N$ at transmitters and receivers, and desired number of signal dimensions $d$ per user.
For generic channel matrices (equivalently for non-degenerate continuously distributed entries) the alignment problem is feasible if and only if
$$
N\geq \frac{d(K+1)}2\,.
$$
\end{theorem}

The present paper is devoted to the proof of this result. We study the problem within the framework of algebraic geometry, and in fact prove a stronger statement on the dimension of the solution set.
Additionally, our upper bound on the dimension of the solution set applies to arbitrary system parameters and implies a necessary condition for interference alignment in this general setting. The proofs of these results are in Section~\ref{s:Proofs}.

Theorem~\ref{t:symmetricFeasibility} suggests an engineering interpretation for
the performance gain from increasing the number of antennas. Depending on
whether $N< d(K+1)/2$ or not, there are two types of performance benefit from
increasing $N$: (1) \emph{alignment gain} or (2) \emph{MIMO gain}. Suppose, for example, that
there are $K=5$ users. If $N=1$, i.e. there is only a single antenna at each
node, then no alignment can be done and only one user can communicate on a
single dimension, giving 1 total degree of freedom (dof). Increasing to $N=2$
antennas allows three users to communicate with one dimension each, giving a total normalized dof $Kd/N=3/2$ (we normalize the total dof by number of antennas $N$). Similarly, increasing to $N=3$ antennas allows all five users to communicate, giving normalized dof $=5/3$. Thus, each increase in $N$ until $N= d(K+1)/2=3$ leads to additional users able to transmit, and a gain of two dimensions per additional antenna; this is \emph{alignment gain}. From here, however, increasing $N$ has a different effect. If we double $N$ to $N=6$, there are still only $5$ users, and each can now transmit along $d=2$ dimensions instead of one, but the normalized dof remains at $5\cdot 2/6=5/3$. The total dof increases at a slower rate, not because more alignment is possible, but simply because more total dimensions are available; this is \emph{MIMO gain}. 

Theorem~\ref{t:symmetricFeasibility} implies that the number of transmit dimensions satisfies $d\leq \frac{2N}{K+1}$. The total normalized dof is therefore $Kd/N=
 \frac{K}{N}\left\lf\frac{2N}{K+1}  \right\rf\leq 2\frac K{K+1}$.
In sharp contrast to the $\frac K2$ total normalized dof achievable for infinitely many parallel channels in~\cite{CJ08}, for the MIMO case we see that at most 2 dof (normalized by the single-user performance of $N$ transmit dimensions) are achievable for any number of users $K$ and antennas $N$.

\subsection{Related work}
The problem we consider, of maximizing dof using vector space strategies for the $K$-user MIMO IC with finite number of transmit and receive antennas, has received significant attention in the last several years.
Cadambe and Jafar \cite{CJ08} considered the problem for $K=3$ users and $N=2$ antennas, and showed that $3/2$ dof was achievable. For more than $3$ users or $N>2$ they assumed an infinite number of parallel channels and applied their main $K/2$ result. Gomadam et al. \cite{GCJ08} posed the problem of determining feasibility of alignment, but left the problem unanswered and proposed a heuristic iterative numerical algorithm.

The main theoretical work to precede the present paper is by Yetis et al. \cite{YGJK10}. Considering the case of a single transmit dimension, $d=1$, they apply Bernstein's Theorem, which requires that each coefficient in a system of polynomial equations is chosen generically. They note that Bernstein's Theorem no longer applies in the case $d>1$, as the equations describing the problem become coupled and coefficients are repeated.
Our approach bypasses the difficulties posed by coupled equations, and thus,
unlike~\cite{YGJK10} our
results do not have the restriction that $d=1$.

Almost all other work has focused on various heuristic algorithms, mainly
iterative in nature (see~\cite{PH09}, \cite{RSL10}, and \cite{SGHP10}). Some have proofs of convergence, but performance guarantees are not available. Schmidt et al. \cite{SGHP10}, \cite{SUH10} study a refined version of the single-transmit dimension problem, where for the case that alignment is possible (as mentioned above, feasibility of alignment is known for the single-transmit case $d=1$) they attempt to choose a good solution among the many possible solutions. Papailiopoulos and Dimakis \cite{PD10} relax the problem of maximizing degrees of freedom to that of a constrained rank minimization and propose an iterative algorithm.

Recently we were notified of independent related work by Razaviyayn et al. \cite{RGL11}. They prove a necessary condition which corresponds to one of our necessary conditions in Theorem~\ref{t:upper-bound}, and they also have a sufficient condition for the special case $d_i=1$ for all $i$ and $M_i=M$, $N_i=N$. However, their necessary condition is not tight, and \cite{RGL11} does not have a sufficient condition for $d>1$ as we do in the present paper (in which we focus on the symmetric square case $M_i=N_i=N$). They have since extended \cite{RGL11-2} their achievability condition to the case where $d_i=d$, and $d$ divides $M_i$ and $N_i$ for each $i$. 

In a different direction of inquiry, Razaviyayn et al. \cite{RSL10} show that
checking the feasibility of alignment for general system parameters is NP-hard.
Note that their result is not in contradiction to ours, since our simple closed-form expression applies only to the fully symmetric case.

We emphasize that in this paper we restrict attention to vector space interference alignment,
where the effect of finite channel diversity can be observed. Interfering signals can also be aligned on the signal scale
using lattice codes (first proposed in~\cite{BPT10}, see also \cite{CJS09}, \cite{EO09}, \cite{MGMK09}), however the understanding of this type of alignment is currently at the stage corresponding to infinite parallel channels in the vector space setting. In other words, essentially ``perfect" alignment is possible due to the infinite channel precision available at infinite signal-to-noise ratios.


\subsection{Interference channel model}
There are $K$ transmitters and $K$ receivers, with transmitter $i$, $1\leq i\leq K$, having $M_i$ antennas and receiver $i$, $1\leq i\leq K$, having $N_i$ antennas. Each receiver $i$ wishes to obtain a message from the corresponding transmitter $i$. The remaining signals from transmitters $j\neq i$ are undesired interference.
The channel is assumed to be constant over time, and at each time-step the input-output relationship is given by
\begin{equation}
  y_i=H_{ii}x_i+\sum_{1\leq j\leq K\atop j\neq i}H_{ij}x_j+z_i\,,\quad 1\leq i\leq K\,.
\end{equation}
Here for each user $i$ we have $x_i\in \C^{M_i}$ and $y_i, z_i\in \C^{N_i}$, with $x_i$ the transmitted signal, $y_i$ the received signal,  and $z_i\sim \CN(0,I_{N_i})$ is additive isotropic white Gaussian noise. The channel matrices are given by $H_{ij}\in \C^{N_i\times M_j}$ for $1\leq i,j\leq K$, with each entry assumed to be independent and with a continuous distribution. We note that this last assumption on independence can be weakened significantly to a basic nondegeneracy condition but we will not pursue this here. For our purposes this means the channel matrices are generic. Each user has an average power constraint, $E(||x_i||^2)\leq P$.

\subsection{Vector space strategies and degrees-of-freedom}
We restrict the class of coding strategies to \emph{vector space} strategies.
In this context degrees-of-freedom (dof) has a simple interpretation as the dimensions of the transmit subspaces, described in the next paragraph. However, we note that one can more generally define the degrees-of-freedom region in terms of an appropriate high transmit-power limit $P\to \infty$ of the Shannon capacity region $C(P)$ normalized by $\log P$ (\cite{CJ08}, \cite{MMK08}). In that general framework, it is well-known and easy to show that vector space strategies give a concrete non-optimal achievable strategy with
rates $$R_i(P)=d_i\log(P)+O(1), \quad 1\leq i\leq K\,.$$ Here $d_i$ is the dimension of transmitter $i$'s subspace and $P$ is the transmit power.

The trasmitters encode their data using vector space precoding.
Suppose transmitter $j$ wishes to transmit a vector $\hat x_j \in \mathbb
C^{d_j}$ of
$d_j$ data symbols. These data symbols are modulated on the subspace
$U_j\subseteq \C^{M_j}$ of dimension~$d_j$, giving the input signal
$x_j=\tilde U_j\hat x_j$, where $\tilde U_j$ is a $M_j \times d_j$ matrix whose
column span is $U_j$. The signal~$x_j$ is received by receiver $i$ through the channel
as $H_{ij}U_j\hat x_j$ 
The dimension of the transmit
space, $d_j$, determines the number of data streams, or degrees-of-freedom,
available to transmitter $j$. With this restriction of strategies, the output is
given by
\begin{equation}
  y_i=H_{ii}\tilde U_i\hat x_i+\sum_{1\leq j\leq K\atop j\neq i}H_{ij}\tilde U_j\hat x_j+z_i\,,\quad 1\leq i\leq K\,.
\end{equation}
The desired signal space at receiver $i$ is thus $H_{ii}U_i$, while the
interference space is given by $\sum_{j\neq i}H_{ij}U_j$, i.e.\ the span of the undesired subspaces as observed by receiver~$i$.

In the regime of asymptotically high transmit powers, in order that decoding can be accomplished we impose the constraint at each receiver~$i$ that the desired signal space $H_{ii}U_i$ is complementary to the interference space $\sum_{j\neq i}H_{ij}U_j$.  Equivalently, there must exist subspaces $V_i$ with $\dim V_i=
\dim U_i$ such that
\begin{equation}\label{e:Generalorthogonality}
 H_{ij}U_j \perp V_i\,, \quad 1\leq i,j \leq K, \quad i\neq j\,,
\end{equation}
 and
\begin{equation}
\dim(\text{Proj}_{V_i}H_{ii}U_i)=\dim U_i\,.
\end{equation}
Here $H_{ij}U_j\perp V_i$ is interpreted to mean that $V_i$ belongs to the dual space $(\C^{N_i})^*$ and $V_i$ annihilates $H_{ij}U_j$. Alternatively, $(V_i)^\dagger H_{ij}U_j=0$, with $V^\dagger$ denoting the Hermitian transpose of $V$. Note that implicitly the transmit dimensions are assumed to satisfy the obvious inequality $d_i\leq \min(M_i,N_i)$.
 If each direct channel matrix $H_{ii}$ has generic (or i.i.d. continuously distributed) entries, then the second condition is satisfied assuming $\dim V_i=d_i$ for each $i$ (this can be easily justified---see \cite{GCJ08} for some brief remarks). Hence we focus on condition \eqref{e:Generalorthogonality}.

The goal is to maximize degrees of freedom, i.e.
choose subspaces $U_1,\dots,U_K$, $V_1,\dots,V_K$ with $d_i\leq\min(M_i,N_i)$ in order to
\begin{align*}
&\text{maximize}\quad d_1+d_2+\dots+d_K \\
&\text{subject to}\quad  H_{ij}U_j \perp V_i\,, \quad 1\leq i,j \leq K, \quad i\neq j\,,
\end{align*}
To this end, it is sufficient to answer the following feasibility question: given number of users $K$, number of antennas $M_1,\dots,M_K$, $N_1,\dots,N_K$, and desired transmit subspace dimensions $d_1,\dots,d_K$, does there exist a choice of subspaces $U_1,\dots,U_K$ and $V_1,\dots,V_K$ with $\dim U_i=\dim V_i=d_i,1\leq i\leq K$, satisfying \eqref{e:Generalorthogonality}?




\subsection{Main results}\label{ss:MainResults}
Our first result contains two parts: (1) it gives the dimension of the variety of solutions in the case that it is generically nonempty, and (2) it gives a necessary condition for feasibility of alignment.

\begin{theorem}\label{t:upper-bound}
Fix an integer $K$ and integers $d_i$, $M_i$, and~$N_i$ for $1 \leq i \leq K$.
For generic channel matrices $H_{ij}$, if there is a feasible strategy, then the
dimension of the variety of such strategies is
\begin{equation*}
\sum_{i=1}^K \big(d_i(N_i - d_i) +  d_i(M_i - d_i)\big)
- \sum_{1\leq i, j\leq K\atop j\neq i} d_i d_j.
\end{equation*}
In particular, if this quantity is negative, then there are no feasible
strategies.

Moreover, in order for there to be feasible solutions, it is necessary that
$$
d_i\leq \min(M_i,N_i)\,,
$$ as well as that
\begin{equation}\label{eq:nec-compl}
d_i + d_j \leq \max\{N_i, M_i\} \qquad
\mbox{for all } 1 \leq i \neq j \leq K,
\end{equation}
and the quantity
\begin{equation}\label{eq:nec-f-A}
t_A := \sum_{i \in A} \big(d_i (N_i - d_i) + d_i (M_i - d_i) \big)
- \sum_{i, j \in A \atop i \neq j} d_i d_j
\end{equation}
is non-negative for all subsets $A \subset \{1, \ldots, K\}$.
\end{theorem}

Theorem~\ref{t:upper-bound} gives
a necessary condition on the existence of solution strategies. In the symmetric
case, this is also a sufficient condition.

\begin{theorem}\label{t:symmetric}
Suppose that $K \geq 3$ and furthermore that $d_i = d$ and $M_i = N_i = N$ for
all users $i$. Then, for generic channel matrices, the space of feasible strategies is
non-empty and has dimension $Kd(2N-(K+1)d)$, if this quantity is non-negative,
and is empty if it is negative.
\end{theorem}

We emphasize that both of these theorems apply only to generic
matrices. This means that there exists an open dense subset of the space of
matrices (in fact the complement of an algebraic hypersurface) on which these
statements hold. In particular, matrices chosen from a non-singular probability
distribution will be sufficiently generic with probability one. On the other
hand, specific matrices, such as $H_{ij} = 0$ for $i \neq j$, may lead to
different solutions.

\begin{corollary}[Symmetric achievable dof]
Under the conditions of Theorem~\ref{t:symmetric}, the maximum normalized dof is given by $$\text{max dof} = \frac{K}{N}\left\lf\frac{2N}{K+1}  \right\rf\leq 2\frac K{K+1}\,.$$
\end{corollary}

\begin{remark}In sharp contrast to the $\frac K2$ result for parallel channels in \cite{CJ08}, for the MIMO case at most 2 dof (normalized by the single-user performance of $N$ transmit dimensions) are achievable.
\end{remark}




%




\section{Proof of main results}\label{s:Proofs}

Some concepts from algebraic geometry will be necessary. For background see the classic reference by Hartshorne \cite{Hartshorne} or the more accessible introduction by Shafarevich \cite{Shaf}.

In algebraic geometry, the basic object of study is the solution set to a system of polynomial equations, called an \emph{algebraic variety} or simply \emph{variety}. The \emph{Zariski topology} is defined by taking the closed sets to be the set of solutions to a system of polynomial equations. Any future reference to closed or open sets is with respect to the Zariski topology. A variety $X$ is \emph{reducible} if it can be written as a union of non-trivial subvarieties $X=X_1\cup X_2$, where $X_1,X_2\neq X$ and $X_1,X_2\neq \varnothing$. A closed set $X$ which is not reducible is \emph{irreducible}. The constituent subsets $X_1,\dots,X_n$ in an irreducible decomposition $X=X_1\cup X_2 \cup \dots \cup X_n$ are called the \emph{components} of $X$. The \emph{dimension} of an irreducible variety $X$ is defined to be the maximum $n$ such that there is a chain of irreducible varieties $Y_0,Y_1,\dots, Y_{n-1}$ satisfying the strict inclusions $\varnothing\subsetneq Y_0\subsetneq Y_1\subsetneq \cdots Y_{n-1}\subsetneq X$.

To represent the strategy space, we will be interested in the \emph{Grassmannian} $G(d,N)$ of $d$-dimensional subspaces of $N$-dimensional affine space $\C^N$. See \cite{Shaf} for more on Grassmannians.
In particular, for each~$i$, the transmit subspace~$U_i$ corresponds to a point
in the Grassmannian, $U_i\in G(d_i,M_i)$, and similarly $V_i\in G(d_i,N_i)$.
The strategy space is thus the product of the Grassmannians,
$\calS=\prod_{i=1}^KG(d_i,M_i) \prod_{i=1}^K G(d_i, N_i)$. Let $\HH=\prod_{i\neq
j} \C^{N_i\times M_i}$ denote the space of all cross channels $H_{ij}$ for
$i\neq j$. Concretely, $h\in \HH$ is a length-$K(K-1)$ tuple of channel matrices $h=(H_{12},H_{13},\dots,H_{K,K-1})$. In the product $\mathcal S \times \mathcal H$, define the incidence variety $\mathcal I\subseteq \calS\times \HH$
to be the set of $(s, h)$ such that $s$ is a feasible strategy for $h$.
Each of $\calS$, $\HH$, and $\mathcal I$ is an algebraic variety.

The following theorem can be thought of as the algebraic geometry analogue of
the rank-nullity theorem from linear algebra (see e.g. Theorem~7 on page~76 of \cite{Shaf}). Given a map
$f\colon X\to Y$, the \emph{fiber} of a point $y$ in $Y$ is the inverse
image of $y$ under the map $f$,  $f\inv(y)=\{x\in X:f(x)=y)\}$. A
\emph{polynomial map} is simply a map whose coordinates are given by polynomials.
\begin{theorem}[Dimension of fibers]\label{t:fibers}
Let $f\colon X\to Y$ be a polynomial map between irreducible varieties. Suppose that $f$ is dominant, i.e.\ the image of $f$ is dense in $Y$. Let $n$ and~$m$ denote the dimensions of $X$ and~$Y$ respectively. Then $m\leq n$ and
\begin{enumerate}
\item $\dim Z\geq n-m$ for any $y\in f(X) \subset Y$ and for any component~$Z$ of the fiber $f\inv (y)$;
\item there exists a nonempty open subset $U\subset Y$ such that $\dim f\inv (y)=n-m$ for $y\in U$.
\end{enumerate}
\end{theorem}
We will apply this theorem to the projections of $\II$ to each of the factors
$\mathcal S$ and~$\mathcal H$.

\begin{lemma}\label{l:dim-x}
$\mathcal I$ is an irreducible variety of dimension
\begin{equation*}
\sum_{i=1}^K \big(d_i (M_i - d_i) + d_i(N_i - d_i)\big)
+ \sum_{1 \leq i,j \leq K \atop i \neq j} (M_i N_j - d_i d_j)
\end{equation*}
\end{lemma}

\begin{proof}
First, we consider the projection
onto the first factor of our incidence variety,
$p\colon\mathcal I \rightarrow \mathcal S$.
For any point $s= (U_1, \ldots,
U_K, V_1, \ldots, V_K)\in \calS$, we claim that the fiber $p\inv (s)$ is  a
linear space of dimension
\begin{equation*}
\dim p\inv (s)=\sum_{1 \leq i, j\leq K \atop i \neq j} M_i N_j - d_i d_j.
\end{equation*}

To see this claim, we give local coordinates to each of the subspaces comprising the solution $s\in \calS$. We write $u^{(i)}_a$ for the $a$th basis element of subspace $U_i$, where $u^{(i)}_a$ has zeros in the first $d_i$ entries except for a 1 in the $a$th entry, and similarly for $v^{(j)}_b$ (this is without loss of generality). 
The orthogonality condition $V_j\perp H_{ji}U_i$ can now be written as the condition $v^{(j)}_b\perp H_{ji} u^{(i)}_a$ for each $1\leq a\leq d_i$ and $1\leq b\leq d_j$. Writing this out explicitly, we obtain
\begin{align*}
0&=v^{(j)}_b\perp H_{ji} u^{(i)}_a
\\&= \sum_{1\leq k\leq M_i\atop 1\leq l\leq N_j} v_b^{(j)}(k)H_{ji}(k,l)u_a^{(i)}(l)
\\
&= \sum_{1\leq k\leq d_i\atop 1\leq l\leq d_j} v_b^{(j)}(k)H_{ji}(k,l)u_a^{(i)}(l)
+\sum_{k> d_i \text{ or }  l>d_j} v_b^{(j)}(k)H_{ji}(k,l)u_a^{(i)}(l)
\\&= H_{ij}(a,b)+\sum_{k> d_i \text{ or }  l>d_j} v_b^{(j)}(k)H_{ji}(k,l)u_a^{(i)}(l)\,.\label{e:orthogonality}
\end{align*}
Note that this equation is linear in the entries of $H_{ji}$. There are $d_i
d_j$ such linear equations, and each one has a unique variable $H_{ji}(a,b)$, so
the equations are linearly independent and each equation reduces the dimension
by 1. The claim follows from the fact that in total there are $\sum_{i\neq
j}d_id_j$ equations.

We have shown that $\mathcal I \rightarrow \mathcal S$ is a vector bundle over
the irreducible variety $\calS$, and thus it is
irreducible. Since $\dim p\inv(s)$ is the same for all $s \in \mathcal S$,
Theorem~\ref{t:fibers} gives the relation
$$\dim\II=\dim\calS+\dim p\inv(s)\,.$$
Since the dimension of~$\mathcal S$ is exactly the first
summation in the lemma statement, this proves the lemma.
\end{proof}

\begin{proof}[Proof of Theorem~\ref{t:upper-bound}]
We now consider the projection onto the second factor~$q\colon\mathcal I \rightarrow \mathcal H$.
If this map is dominant (i.e., generically the alignment problem is feasible), then by
Theorem~\ref{t:fibers} the fiber $q\inv (h)$ for a generic $h\in\mathcal
H$ has dimension \begin{equation}\label{e:dimqproj}
\dim q\inv (h)=\dim \mathcal I - \dim \mathcal H\,.\end{equation}
Since $\mathcal H$ has dimension
equal to $\sum_{1 \leq i, j \leq K \atop i \neq j} M_i N_j$, then
Lemma~\ref{l:dim-x} gives us the dimension in the statement of the theorem. Moreover,
if the quantity in \eqref{e:dimqproj} is negative, then the fiber $q\inv (h)$ at a generic point must
be empty. But the set of solutions to the tuple of channel matrices $h$ is given by $p(q\inv (h))$, so for generic channel matrices this means there is no feasible strategy.

Now we turn to the other necessary conditions for the existence of a solution. The first necessary condition $d_i\leq \min(M_i,N_i)$ is obvious. Next,
suppose that
$d_i + d_j > N_i \geq M_j$ for some $i$ and~$j$. Since $H_{ij}$ is a generic
$N_i \times M_j$ matrix, its nullspace will be trivial. Thus, $H_{ij} U_j$ will
be a $d_j$-dimensional vector space. Since $d_i + d_j > N_i$, the vector spaces
$H_{ij} U_j$ and $V_i$ cannot be orthogonal. If $d_i + d_j > M_j \geq N_i$, then the argument is
similar, but with the roles of $U_j$ and $V_i$ reversed.


Finally, any feasible strategy for the full set of $K$ transmitters and
receivers, will, in particular be feasible for any subset. Therefore, a
necessary condition for a general set of channel matrices to have a feasible
strategy is that the same is true for any subset of the pairs. Since the number
$t_A$ is the dimension of the variety of solutions when restricted just to the
transmitters and receivers indexed by $i \in A$, then $t_A$ must be non-negative in order to have a feasible strategy.
\end{proof}

Now, we make the assumption that
$N_i=M_i=N$ and $d_i=d$ for all $1\leq i\leq K$, and also that $K\geq 3$, and we
wish to prove a sufficient condition for the existence of a feasible strategy in
Theorem~\ref{t:symmetric}.
The following lemma reduces the problem of showing that almost all channel
tuples $h\in \HH$ have a solution to finding the dimension of the solution set
for a single channel tuple $h\in \HH$. Recall that $q$ is the projection of the
incidence variety $\II$ onto the second factor, and that $q$ being dominant
means that its image is dense in $\HH$, i.e.\ generic channel matrices have a solution.
\begin{lemma} \label{l:ZariskiCotDim2}
Suppose that there exists $h \in \mathcal H$ such that the dimension of
$q^{-1}(h)$ is at most $Kd(2N-(K+1)d)$. Then $q$ is dominant.
\end{lemma}
\begin{proof}
Let $h\in \HH$ be a point such that $q\inv(h)$ has at most the stated dimension. Let
$\ZZ_0=q(\II)$ be the projection of $\II$ onto the second factor, and
let $\ZZ$ denote the closure of~$\ZZ$. By these definitions, the projection
$q\colon\II\to\ZZ$ is dominant. Now, part~1 of Theorem~\ref{t:fibers} (dimension of
fibers) gives
$$
\dim q\inv (h)\geq \dim \II-\dim \ZZ\,,
$$
from which it follows that
$$\dim \ZZ\geq \dim \II-\dim q\inv(h)=\dim \HH\,.$$
But $\ZZ\subseteq \HH$, so equality of dimensions and irreducibility of $\HH$
implies $\ZZ=\HH$ (see e.g. \cite[Thm.~1, pg.~68]{Shaf}), or, in other words, that $q\colon\II\to\HH$ is dominant.
\end{proof}


Using Lemma~\ref{l:ZariskiCotDim2}, proving Theorem~\ref{t:symmetric} requires only that we find a tuple of channels $h\in \HH$ so that the set of solutions has the correct dimension. This is provided by the following lemma.
\begin{lemma}
Suppose that $K \geq 3$ and furthermore that $d_i = d$ and $M_i = N_i = N$ for
all users $i$. If $Kd(2N-(K+1)d)\geq 0$, then there exists $h\in \HH$ such that the dimension of $\dim q\inv (h)$ is at most $Kd(2N-(K+1)d)$.
\end{lemma}
\begin{proof}
We consider the point $s_0\in \calS$ where each $U_i$ and~$V_i$ is spanned by
the first $d$ standard basis vectors. Therefore, the set of channel matrices for
which $s_0$ is a valid strategy are those $H_{ij}$ such that
$$
\left(\begin{matrix} \mathrm{Id}_d \\ 0\end{matrix}\right)^T H_{ij}
\left(\begin{matrix} \mathrm{Id}_d \\ 0 \end{matrix} \right) = 0,
$$
for $i$ and $j$ distinct integers between $1$ and~$K$. Here $\mathrm{Id}_n$ denotes the $n\times n$ identity matrix.
It is clear that this implies that the upper left corner of $H_{ij}$ must be
zero, and thus we can write it in the form
$$
H_{ij}=\left(
\begin{matrix}
0 & F_{ij} \\
G_{ij} & \tilde{H}_{ij}
\end{matrix}\right)\,,
$$
where $F_{ij}$, $G_{ij}$, and $\tilde{H}_{ij}$ can be any matrices of size
$d \times (N-d)$, $(N-d) \times d$, and $(N-d) \times (N-d)$ respectively.
In a moment we will
specify $F_{ij}$ and $G_{ij}$, but for now we assume they are fixed, but
arbitrary.

We now investigate the set of solution strategies for these fixed channel matrices~$H_{ij}$.
In local coordinates around the strategy~$s_0$, the vector spaces $U_i$
and~$V_i$ can be written as column spans as follows:
$$
U_i = \operatorname{col span} \left( \begin{matrix} \mathrm{Id} \\ \bar U_i\end{matrix}\right), \qquad
V_i = \operatorname{col span} \left( \begin{matrix} \mathrm{Id} \\ \bar V_i\end{matrix}\right)\,,
$$
where $\bar U_i$ and $\bar V_i$ are $(N-d) \times d$ matrices of variables.
In order to satisfy the orthogonality condition, we need that
$$
\left( \begin{matrix} \mathrm{Id} \\ \bar V_i\end{matrix}\right)^T
H_{ij}
\left( \begin{matrix} \mathrm{Id} \\ \bar U_j\end{matrix}\right)
= \bar{V}_i^T G_{ij} + F_{ij} \bar{U}_j + \bar{V}_i^T \tilde{H}_{ij} \bar{U}_j
= 0.
$$
We linearize this problem by dropping the final, quadratic term:
\begin{equation}\label{e:ZariskiCot}
 \bar{V}_i^T G_{ij} + F_{ij} \bar{U}_j=0\,,\qquad 1\leq i,j\leq K
\,.
\end{equation} In algebraic
geometry, the vector space defined by the linear equations \eqref{e:ZariskiCot} is known as
the Zariski cotangent space, and its dimension gives an upper bound on the
dimension of the variety at the given point~\cite[Thm. 9.6.8(ii)]{CLO}. Hence we focus on the problem of computing the dimension of the set of solutions $(\bar U_j,\bar V_i)$ to \eqref{e:ZariskiCot}.

We now give our construction of the
matrices $F_{ij}$ and~$G_{ij}$ and find the dimension of the Zariski cotangent space.
We separate the construction into two cases: (1) $K$ is odd, (2) $K$ is even.

\paragraph{Case 1: $K$ is odd.} This case is relatively straightforward.
Recall that $F_{ij}$ is of size $d\times (N-d)$ and $G_{ij}$ is of size $(N-d)\times d$. We write
\begin{equation}\label{e:KoddFG}
F_{ij}=
\left(\begin{matrix}
A_{ij}(\frac{K-1}2+1) &A_{ij}(\frac{K-1}2+2) & \cdots &  A_{ij}(K-1) & 0
\end{matrix} \right)\,, \quad
G_{ij}=
\left(\begin{matrix}
A_{ij}(1)  \\
A_{ij}(2)  \\
\vdots \\
A_{ij}(\frac{K-1}2) \\
 0
\end{matrix} \right)\,,
\end{equation}
where each $A_{ij}(k)$ is a block of size $d\times d$ to be defined shortly, and the rightmost zero in $F_{ij}$ is of size
$d \times (N-d\frac{K+1}2)$ while
the bottom zero in $G_{ij}$ is a block of zeros of size
$(N-d\frac{K+1}2)\times d$.
Note that the assumption $2Kd(N-d)\geq K(K-1)d^2$ is equivalent to $N\geq d\frac{K+1}2$, so the specification of $F_{ij}$ and $G_{ij}$ above makes sense.

Let
\begin{equation}\label{e:Aij}
A_{ij}(k)=
\begin{cases}
\Id_{d}\quad &\text{if } k=i-j\\
0\quad &\text{otherwise}
\end{cases}\,,
\end{equation} where addition of indices is modulo $K$. 
Now write
\begin{equation}\label{e:KoddUV}
  \bar U_i=\left( \begin{matrix}
    X_i(1)\\
    X_i(2)\\
    \vdots \\
    X_i({\frac{K-1}2}) \\
    X_i
  \end{matrix}\right)\,,\quad
  \text{and} \quad
    \bar V_i=\left( \begin{matrix}
    Y_i(1)\\
    Y_i(2)\\
    \vdots \\
    Y_i({\frac{K-1}2}) \\
    Y_i
  \end{matrix}\right)\,,
\end{equation}
where $X_i(t),Y_i(t)$, $1\leq t\leq \frac{K-1}2$ are $d\times d$ blocks of variables from $\bar U_i,\bar V_i$, respectively, and $X_i$, $Y_i$ are blocks of size $[N-d\frac{K+1}2]\times d$ containing the remaining variables.

With this notation and choice of matrices, the equations \eqref{e:ZariskiCot} defining the Zariski cotangent space read
\begin{equation}
  X_i(t)=0, \quad 1\leq t\leq \frac{K-1}2, \quad \text{and} \quad Y_i(t)=0, \quad 1\leq t\leq \frac{K-1}2.
\end{equation}
Considering these equations for all $i$, $1\leq i\leq K$, we see that the codimension is $d^2 2K \left(\frac{K-1}2\right)=K(K-1)d^2$. Thus
the set of solutions to the equations \eqref{e:ZariskiCot} have dimension $2Kd(N-d)-K(K-1)d^2$.\\


\paragraph{Case 2: $K$ is even.} The idea behind the argument is the same as in case 1 but the details are slightly more involved.
Put
\begin{equation}\label{e:KoddFG}
F_{ij}=
\left(\begin{matrix}
A_{ij}(\frac{K+2}2+1) &A_{ij}(\frac{K+2}2+2) & \cdots &  A_{ij}(K-1) & \hat F_{ij} & 0
\end{matrix} \right),
\quad
G_{ij}=
\left(\begin{matrix}
A_{ij}(4)  \\
A_{ij}(5)  \\
\vdots \\
A_{ij}(\frac{K+2}2) \\
 \hat G_{ij}\\
0
\end{matrix} \right),
\end{equation}
where the matrices $A_{ij}(k)$ are defined above \eqref{e:Aij}, $\hat F_{ij}$ is a block of size $d\times \lrc{\frac{3d}2}$ when $j$ is odd and size $d\times \lrf{\frac{3d}2}$ when $j$ is even, $\hat G_{ij}$ is a block of size $\lrc{\frac{3d}2}\times d$ when $i$ is odd and size $\lrf{\frac{3d}2}\times d$ when $i$ is even, and any remaining entries are zero.
Note that the assumption $2Kd(N-d)\geq K(K-1)d^2$ is equivalent to $N-d-d\frac{K-4}2\geq \lrc{\frac{3d}2}$, so the specification of $F_{ij}$ and $G_{ij}$ makes sense.

We write
\begin{equation}
  \bar U_i=\left(\begin{matrix}
U^0_{i} \\
X_i\\
\bar X_i
\end{matrix} \right)\,, \quad \text{and}\quad   \bar V_i=\left(\begin{matrix}
V^0_{i} \\
Y_i\\
\bar Y_i
\end{matrix} \right)\,,
\end{equation}where $U^0_i,V^0_i$ are of size $d\frac{K-4}2\times d$, $X_i,Y_i$ are of size $\lrc{\frac{3d}2}\times d$ for odd values of $i$ and of size $\lrf{\frac{3d}2}\times d$ for even values of $i$ , and $\bar X_i,\bar Y_i$ contain the remaining variables (if any).
Now, exactly as in case 1 above, the choice of $F_{ij},G_{ij}$ forces $U_i^0=V_i^0=0$.

It remains to specify the matrices $\hat F_{ij},\hat G_{ij}$. Let
\begin{equation}
    \hat F_{ij}=\left(\begin{matrix}
A_{ij}(3) & B_{ij}
\end{matrix} \right)\,,
\quad \text{and}\quad
\hat G_{ij}=\left(\begin{matrix}A_{ij}(1) \\
B_{ij}^T
\end{matrix} \right)\,,
\end{equation}
where again $A_{ij}(k)$ is defined in \eqref{e:Aij}.
The matrices $B_{ij}$ are either $d\times \lrc{\frac d2}$ or$d\times \lrf{ \frac d2}$ depending on the indices;
$B_{ij}$ is given by
\begin{equation}\label{e:Bij}
B_{ij}(k)=
\begin{cases}
\left(\begin{matrix} \Id_{\dc} \\ 0\end{matrix} \right) & \text{ if } i \text{ is even and }j=i+1 \text{ or } i\text{ is odd and }j=i+3 \\
\left(\begin{matrix} \Id_{\df} \\ 0\end{matrix} \right) & \text{ if } i \text{ is odd and }j=i+1 \\
\left(\begin{matrix} 0\\ \Id_{\dc} \end{matrix} \right) & \text{ if } i \text{ is odd and }j=i+2 \\
\left(\begin{matrix} 0 \\ \Id_{\df} \end{matrix} \right) & \text{ if } i \text{ is even and }j=i+3\text{ or } i\text{ is even and }j=i+2 \\
0\quad &\text{otherwise}
\end{cases}\,,
\end{equation} where the 0 in \eqref{e:Bij} denotes is a block of zeros of appropriate size to ensure that $B_{ij}$ is rectangular with $d$ rows. Here, again, addition of indices is modulo $K$.

Let
  \begin{equation}
X_i= \left(\begin{matrix}
    \begin{matrix}
      X_i^1 & X_i^2
    \end{matrix}
    \\ X_i^3
  \end{matrix}\right)\,,
\quad \text{and}\quad
  Y_i^T=\left(\begin{matrix}
    \begin{matrix}
      Y_i^1 \\ Y_i^2
    \end{matrix}
    & Y_i^3
  \end{matrix}\right)\,,
\end{equation}
 where for \emph{even} values of $i$ the blocks $X_{i+1}^3,Y_{i}^1,Y_{i+1}^2$ are $\lrc{\frac d2}\times d$, $X_i^3,Y_{i}^2,Y_{i+1}^1$ are $\lrf{\frac d2}\times d$, $X_{i}^1,X_{i+1}^2,Y_{i+1}^3$ are $d\times \lrc{\frac d2}$, and $X_{i}^2,X_{i+1}^1,Y_{i}^3$ is $d\times \lrf{\frac d2}$.
Then our linear equation \eqref{e:ZariskiCot} implies
\begin{align*}0=Y_i^T \hat G_{ij}+\hat F_{ij} X_j&=\left(\begin{matrix}
  Y_i^1 \\ Y_i^2
\end{matrix}\right)A_{ij}(1)+Y_i^3B_{ij}^T
\\ &\qquad + A_{ij}(3)\left( \begin{matrix}
    X_j^1 & X_j^2
\end{matrix}\right)+B_{ij}X_j^3\,.
\end{align*}
For each even value of $i$  we end up with the equations
\begin{alignat*}{2}
0&=\left(\begin{matrix}Y_{i+3}^1+X^3_{i} \\ Y_{i+3}^2\end{matrix}\right),\quad 0 =\left(\begin{matrix}Y_{i+4}^1+X^3_{i+1}\\ Y_{i+4}^2 \end{matrix}\right),\quad 0=\left(\begin{matrix}Y_{i+3}^1\\ Y_{i+3}^2+X^3_{i+1} \end{matrix}\right), \\
0&=\left(\begin{matrix}X_{i}^1+Y^3_{i+1} & X_{i}^2\end{matrix}\right),\quad 0=\left(\begin{matrix}X_{i}^1& X_{i}^2 +Y^3_{i+2}\end{matrix}\right),\quad 0=\left(\begin{matrix}X_{i+1}^1& X_{i+1}^2+Y^3_{i+2} \end{matrix}\right),
\end{alignat*}
which implies that all variables appearing here are zero, i.e. $X_i,Y_i=0$ for each $i$.

This proves that precisely $K(K-1)d^2$ entries of $V_i,V_i$ must be zero for this choice of matrices $F_{ij}, G_{ij}$, finishing case 2 and completing the proof of the lemma.
\end{proof}

\section*{Acknowledgment}
We thank Bernd Sturmfels for insightful discussions and the authors of \cite{RGL11} for sharing their related manuscript with us.

\nocite{BT09}

\bibliographystyle{ieeetr}
\bibliography{BIB}

\end{document}